\documentclass[10pt]{amsart}

\usepackage{amsmath,amssymb,amsthm,amscd,amsfonts}
\usepackage{fullpage}
\usepackage{graphicx}
\usepackage{stmaryrd}
\usepackage{float}
\usepackage{hyperref}

\usepackage{algorithm}							
\usepackage[noend]{algpseudocode}					
\usepackage{caption}
\usepackage[margin=1in]{geometry}
\usepackage{url}

\usepackage{mathtools}
\DeclarePairedDelimiter{\ceil}{\lceil}{\rceil}
\DeclarePairedDelimiter\floor{\lfloor}{\rfloor}

\usepackage{amssymb,latexsym}
\usepackage{amsmath,amsfonts,amsthm,amscd,amsxtra,enumerate}
\usepackage{mathtools}
\usepackage{caption}
\usepackage[all]{xy}
\usepackage{amsmath,latexsym,amssymb, enumerate, amsthm, epsfig, hyperref} 
\usepackage[margin=1in]{geometry}
\usepackage{epsfig}
\usepackage{pgf,tikz}

\usepackage[utf8]{inputenc}

\newtheorem{conjeture}{ Conjecture}[section]
\newtheorem{theorem}[conjeture]{ Theorem}
\newtheorem{lemma}[conjeture]{ Lemma}
\newtheorem{corollary}[conjeture]{ Corollary}
\newtheorem{proposition}[conjeture]{ Proposition}
\theoremstyle{definition}
\newtheorem{remark}[conjeture]{ Remark}

\newtheorem{definition}[conjeture]{ Definition}

\newtheorem{example}[conjeture]{ Example}

\begin{document}

\title{Palindromes in two-dimensional Words}

\author{Kalpana Mahalingam, Palak Pandoh  }

\address
	{Department of Mathematics,\\ 
	 Indian Institute of Technology Madras, 
	  Chennai, 600036, India}
	  \email{kmahalingam@iitm.ac.in,palakpandohiitmadras@gmail.com }
	

%
%

\begin{abstract} 
A two-dimensional ($2$D) word is a $2$D palindrome 
if it is equal to its reverse and it is an HV-palindrome if all its columns and rows are $1$D palindromes. We study some combinatorial and structural properties of HV-palindromes and its comparison with $2$D palindromes. We investigate the maximum number
number of distinct non-empty HV-palindromic sub-arrays in any finite $2$D word, thus, proving the conjecture given by Anisiua et al. We also find the least number of HV-palindromes in an infinite $2$D word over a finite alphabet size $q$.

\end{abstract}
\maketitle
\section{Introduction}
Palindromes are extensively studied in $1$-dimensional words by several authors (See \cite{pal3,mc20,palrich,2015arXiv150309112G}). There is an increasing interest in the combinatorial properties of palindromes in mathematics, theoretical computer science, and biology. The notion of $\theta$-palindrome was defined in \cite{Anne20} and studied in \cite{watson} and \cite{LUCA2006282} 
independently. Some properties that link the
 palindromes to classical notions such as that of primitive words are present in \cite{MR2658256}. Due to their symmetrical properties, this concept was generalized to two-dimension. Such a construction has significance in detecting bilateral symmetry of an image and face recognition technologies (\cite{id1,id2}).

Identifying palindromes in  arrays dates back to $1994$, when authors in \cite{MR1272521}, described  an array to be  a palindrome if all its rows and columns are $1$D palindromes. These structures are referred to as HV-palindromes in \cite{mc21}, where H and V stand for horizontal and vertical respectively. It was much later in $2001$ when Berth\'{e} et al., \cite{Vullion2001}  formally defined a $2$D palindrome to be an array which is equal to its reverse. It can be easily observed that an array whose all rows and columns are $1$D  palindromes is equal to its reverse. Hence, HV-palindromes is a sub-class of $2$D palindromes. Recently, there has been a rise in research that deals with the concept of $2$D palindromes. A relation between $2$D palindromes and $2$D primitive words was studied in \cite{ms}.  An algorithm for finding maximal $2$D palindromes was given in \cite{Sara2016}. The maximum and the least number of $2$D palindromic sub-arrays in a given array was studied in  \cite{mc21,max} and \cite{tcs} respectively. 

The main idea of this paper is to study the  structure of a special type of $2$D word called an HV-palindrome.  The motivation of studying this structure came from the conjecture mentioned in \cite{mc21} which speculates about the maximum number of HV-palindromic sub-arrays in a $2$D word of size $(2,n)$ for a given $n$. We settle this conjecture in affirmative and generalize the result to $2$D words of larger sizes.

This paper is organized as follows.
Section $3$ deals with the characterization of a word to be an HV-palindrome. Further, the $2$D words whose all $2$D palindromes are HV-palindromes are also characterized. In Section $4$, we count the number of possible $2$D palindromes and HV-palindromes for a given array size and investigate the number of $2$D palindromic and HV-palindromic conjugates of a $2$D word. In Section $5$, we find the upper and lower bounds on the  number of HV-palindromic sub-arrays in a $2$D word. We end the paper with some concluding remarks.

\section{Basic definitions and notations}\label{sec2}
An alphabet $\Sigma$ is a finite non-empty set of symbols. A $1$D word is defined to be a sequence of symbols from $\Sigma$. $\Sigma^{*}$ denotes the set of all words over $\Sigma$ including the empty word $\lambda$. $\Sigma^{+}=\Sigma \setminus \lambda$. The length of a word  $w \in \Sigma^{*}$ is the number of symbols in a word and is denoted by $|w|$. The reversal of $w=a_{1}a_{2} \cdots a_{n}$ is defined to be a string $w^{R}=a_{n} \cdots a_{2} a_{1}$ where $a_{i} \in \Sigma$. $Alph(w)$ is  the set of all sub-words of $w$ of length $1$. A word $w$ is said to be a palindrome if $w=w^{R}$. The concepts of prefix, suffix, primitivity, and conjugates are as usual. For all other concepts in formal language theory and combinatorics on words, the reader is referred to \cite{Hopcroft69,Lothaire97}.
 
 \subsection{Two-dimensional arrays} 
 
 A two-dimensional word  $w=[w_{i,j}]_{1 \leq i \leq m,1 \leq j \leq n}$ over $\Sigma$ of size $(m,n)$ is defined to be a two-dimensional rectangular array of letters. If both $m$ and $n$ are infinite, then $w$ is an infinite $2$D word. A factor of $w$ is a sub-array/sub-word of $w$. In the case of $2$D words, an empty word is a word of size $(0,0)$, and we use the notation $\lambda$ to denote such a word. The set of all $2$D words including the empty word $\lambda$ over $\Sigma$ is denoted by $\Sigma^{**}$ whereas, $\Sigma^{++}$ is the set of all non-empty $2$D words over $\Sigma$. Note that, the words of size $(m,0)$ and $(0,m)$ for $m>0$ are not defined.

\begin{definition}
Let $u=[u_{i,j}]$ and $v=[v_{i,j}]$ be two $2D$ words over $\Sigma$ of size $(m_{1},n_{1})$ and $(m_{2},n_{2})$, respectively.

\begin{enumerate}
\item The column concatenation of $u$ and $v$ $($denoted by $\obar)$ is a partial operation, defined if $m_{1}=m_{2}=m$, and it is given by 
$$u \obar v=
\begin{matrix}
u_{1,1} & \cdots & u_{1,n_{1}} & v_{1,1} & \cdots & v_{1,n_{2}}  \\
\vdots & \ddots  & \vdots & \vdots &  \ddots & \vdots \\
u_{m,1} & \cdots & u_{m,n_{1}} & v_{m,1} & \cdots & v_{m,n_{2}}
\end{matrix}$$
The column closure of $u$ $($denoted by $u^{* \obar})$ is defined as $u^{* \obar}=\bigcup_{i \geq 0} u^{i \obar}$ where $u^{0 \obar}=\lambda,u^{1 \obar}=u,u^{n \obar}=u \obar u^{(n-1) \obar}$.
As $n \rightarrow \infty$, we get $u^{\infty \obar}$.

\item The row concatenation of $u$ and $v$ $($denoted by $\ominus)$ is a partial operation defined if $n_{1}=n_{2}=n$, and it is given by 
$$u \ominus v=
\begin{matrix}
u_{1,1} & \cdots & u_{1,n}  \\
\vdots &  \ddots & \vdots \\
u_{m_{1},1} & \cdots & u_{m_{1},n} \\
v_{1,1} & \cdots & v_{1,n} \\
\vdots &  \ddots & \vdots \\
v_{m_{2},1} & \cdots & v_{m_{2},n}
\end{matrix}$$
The row closure of $u$ $($denoted by $u^{* \ominus})$ is defined as $u^{* \ominus}=\bigcup_{i \geq 0} u^{i \ominus}$ where $u^{0 \ominus}=\lambda,u^{1 \ominus}=u,u^{n \ominus}=u \ominus u^{(n-1) \ominus}$.
As $n \rightarrow \infty$, we get $u^{\infty \ominus}$. \end{enumerate} 
\end{definition}

In \cite{amir1998}, a prefix of a $2$D word $w$ is defined to be a rectangular sub-block that contains one corner of $w$, whereas suffix of $w$ is defined to be a rectangular sub-array that contains the diagonally opposite corner of $w$. However, we consider prefix/suffix of a $2$D word $w$ to be a rectangular sub-array that contains the top-left/bottom-right corner of $w$. This was formally defined in \cite{ms} as follows.

\begin{definition}
Given $u \in \Sigma^{**}$, $v \in \Sigma^{**}$ is said to be a prefix of $u$ $($respectively, suffix of $u)$, denoted by $v \leq_{p} u$  $($respectively $v \leq_{s} u)$ if $u=(v \ominus x) \obar y$ or $u=(v \obar x) \ominus y$ $($respectively, $u=y \obar (x \ominus v)$ or $u=y \ominus (x \obar v))$ for $x,y \in \Sigma^{**}$.

\end{definition}
A \textit{border} of a $2$D word is a sub-array  that occurs as both the prefix and suffix.
 \begin{definition}
Let  $w = [w_{i,j}]_{1\le i \le m, 1\le j\le n}$ be a $2$D word. The reverse and transpose of $w$, denoted by $w^R$ and $w^T$ respectively are defined as 
\begin{center}
$w^R= \begin{matrix}
w_{m,n} & w_{m,n-1} & \cdots & w_{m,1} \\
w_{m-1,n} & w_{m-1,n-1} & \cdots  & w_{m-1,1} \\
\vdots & \vdots & \ddots  & \vdots \\
w_{1,n} & w_{1,n-1} & \cdots & w_{1,1}
\end{matrix}\;\;\;\;\;\;
w^{T}=
\begin{matrix}
w_{1,1} & w_{2,1} & \cdots  & w_{m,1} \\
w_{1,2} & w_{2,2} & \cdots  & w_{m,2} \\
\vdots & \vdots & \ddots & \vdots \\
w_{1,n} & w_{2,n} & \cdots  & w_{m,n}
\end{matrix}$
\end{center}
 \end{definition}
If $w=w^{R}$, then $w$ is said to be a \textit{two-dimensional $(2$D$)$ palindrome} (\cite{Vullion2001,Sara2016}). 
\begin{example}\label{e1}
Let $\Sigma=\{a,b,c\}$ and  
$$
w=\begin{matrix}
a & b & c & a \\
b & c & c & a \\
a & c & c & b\\
a & c & b & a 
\end{matrix}
$$ 
As $ w=w^R$, so $w$ is a $2$D palindrome.
\end{example}
Note that the rows and columns of a $2$D palindrome are not always $1$D palindromes. Also, if all columns and rows of a finite $2$D word are palindromes, then the word is a $2$D palindrome. Such palindromes are referred to as \textit{HV-palindromes} in \cite{mc21}. 
\begin{example}\label{e2}
Let $\Sigma=\{a,b\}$ and   
$$u=\begin{matrix}
a & b  & a \\
b & c & b \\
a & b & a 
\end{matrix}$$
As every row and column of $u$ is a $1$D palindrome, so $u$ is an HV-palindrome.
\end{example}
We recall the notion of horizontal and vertical palindromes from \cite{max}.
\begin{definition} 
The \textit{horizontal palindromes} of a $2$D word $w$ are the palindromic factors of $w$ of size $(1,j)$  where $j\geq 1$ and $ vertical\; palindromes$ of $w$ are the palindromic factors of $w$ of size $(i,1)$ where $i\geq 2$. 
\end{definition}
The palindromes of size $(1,1)$ are trivial and rest are non-trivial. Note that, all horizontal and vertical palindromes are HV-palindromes. We now recall the notion of the center of a $2D$ word as defined in \cite{Sara2016}. 
\begin{definition}
Center is the position that results in an equal number of columns to the left and right, as well as an equal number of rows above and below, i.e. in a word of size $(m,n)$, if $m$ and $n$ are odd, then the center is at location $(\ceil{\frac{m}{2}},\ceil{\frac{n}{2}})$. If $m$ or/and $n$ is even, the center is in between rows or/and columns respectively. 
\end{definition}
\begin{example}
The center of $w$ in Example \ref{e1} is in between the rows and columns of the sub-array $
     cc\ominus cc$.
However, for the word $u$ from Example \ref{e2}, the center of $u$ is the sub-word \textit{c}.
\end{example}

Throughout the paper, by `the number of palindromes', we mean the number of non-empty distinct palindromic factors. For more information pertaining to two-dimensional word concepts, we refer the reader to \cite{Vullion2001,Taylor7,Sara2016,giam1997}.

\section{Structure of an HV-palindrome}
By definition, if a $2$D word $w$ of size $(m,n)$ is a palindrome, then $w=w^R$ and $w$ admits two symmetries namely identity and $180^{\circ}$ rotation. Hence, given a $2$D palindrome $w=[w_{i,j}]$ of size $(m,n)$, let $p_1$ be the prefix of size $(\floor{\frac{m}{2}},n)$ and $p_2= w_{\ceil{\frac{m}{2}},1}w_{\ceil{\frac{m}{2}},2}\cdots w_{\ceil{\frac{m}{2}},n}$. Then, $w$ is of the form 

\begin{enumerate}
    \item   $p_1 \ominus p_1^R $, if $m$ is even.   
    \item $ p_1 \ominus p_2\ominus p_1^R$, if $m$ is odd.
\end{enumerate}
In addition to the symmetries in a $2$D palindrome, an HV-palindrome is preserved under reflections about horizontal and vertical axis containing the center of the word. Due to such symmetrical properties, we give the exact structure of an HV-palindrome. 
\begin{theorem}\label{3}
\text{$($Structure theorem of HV-palindromes$)$}\\ Given an HV-palindrome $w=[w_{i,j}]$  of size $(m,n)$, let $u=u_1\ominus u_2 \ominus\cdots\ominus u_{\floor{\frac{m}{2}}}$ be the prefix of $w$ of size $(\floor{\frac{m}{2}},\floor{\frac{n}{2}})$, \;$v={u_1}^R\ominus {u_2}^R \ominus\cdots\ominus {u_{\floor{\frac{m}{2}}}}^R$, \;$p_1= w_{1,\ceil{\frac{n}{2}}}w_{2,\ceil{\frac{n}{2}}}\cdots w_{\floor{\frac{m}{2}},\ceil{\frac{n}{2}}}$ and $p_2= w_{\ceil{\frac{m}{2}},1}w_{\ceil{\frac{m}{2}},2}\cdots w_{\ceil{\frac{m}{2}},\floor{\frac{n}{2}}}$. Then, $w$ is of one of the following forms. 
\begin{center}
\begin{tabular}{|c|c|c|c|c|c|}
\hline
     $m$ & $n$ & \;Structure of $w$\;& $m$ & $n$ & Structure of $w$ \\
     \hline\hline
    
     \;even\; & \;even\;& $\begin{matrix} u &v\\v^R&u^R \end{matrix}$ & \;even\; &\;odd\;&$\begin{matrix} u & p_1^T&v\\v^R&(p_1^T)^R&u^R \end{matrix}$
     \\
     \hline
        
\;odd\; &\;even\; &
      $\begin{matrix} u &v\\p_2& p_2^R\\ v^R&u^R \end{matrix}$ & \;odd\; &\;odd\; & $\begin{matrix} u & p_1^T&v\\p_2& x&p_2^R\\ v^R&(p_1^T)^R&u^R\end{matrix}$ \\

      \hline
\end{tabular}

\end{center}
where $x=w_{\ceil{\frac{m}{2}},\ceil{\frac{n}{2}}}. $
\end{theorem}
\begin{proof}
We give the proof of the case when $m$ and $n$ are both even and the rest of the cases follow similarly.
Let $w=[w_{i,j}]$ be an HV-palindrome of size $(m,n)$ and $u=u_1\ominus u_2 \ominus \cdots \ominus u_{\frac{m}{2}}$ be the prefix of $w$ of size $(\frac{m}{2},\frac{n}{2})$. Now, as every row of $w$ is a palindrome, $u_1u_1^R\ominus u_2u_2^R \ominus \cdots \ominus u_{\frac{m}{2}}u_{\frac{m}{2}}^R=u\obar v$ is the prefix of $w$ of size $(\frac{m}{2},n)$. Also, as every column of $w$ is a palindrome, then  $w= (u \obar v)\ominus (v^R\obar u^R)$. 
\end{proof}

An HV-palindrome of size $(m,n)$ can be constructed from an HV-palindrome of size $(m+1,n+1)$
by removing its $\ceil{\frac{m}{2}}^{th}$ row and $\ceil{\frac{n}{2}}^{th}$ column. Such a construction in the case of a $2$D palindrome is given in \cite{ms}. We also observe the following result. 

\begin{lemma}\label{l3}
If $w$ is an HV-palindrome of size $(m,n)$, then the word obtained by the removal of first and last $k$ rows of $w$, for $1\leq k\leq \floor{\frac{m}{2}}$ and first and last $r$ columns of $w$, for $1\leq r\leq \floor{\frac{n}{2}}$ is an HV-palindrome of size $(m-2k,n-2r)$. This result is also true in the case of $2$D palindromes.
\end{lemma}


\subsection{Characterization of an HV-palindrome}

In this section, we first give two necessary and sufficient conditions for a $2$D word to be an HV-palindrome and then give a characterization of  $2$D words such that all of their palindromic sub-words are HV-palindromes i.e., words with no non-HV-palindromic sub-arrays.

\begin{proposition}\label{1}
Let $w=w_1\ominus w_2 \cdots \ominus w_m= u_1\obar u_2 \cdots \obar u_n$ be a $2$D word of size $(m,n)$, where $w_i$ and $u_j$ is the $i^{th}$ row and $j^{th}$ column of $w$ respectively. Then, $w$ is an HV-palindrome if and only if $w_i=w_{m-i+1}$ for $1\leq i\leq \floor {\frac{m}{2}}$ and $u_j=u_{n-j+1}$ for $1\leq j\leq \floor {\frac{n}{2}}$.
\end{proposition}
\begin{proof}
Let $w=w_1\ominus w_2 \cdots \ominus w_m= u_1\obar u_2 \cdots \obar u_n$ be an HV-palindrome. This implies each $w_i$ and $u_j$ is a palindrome. Every row of $w$ is a palindrome if and only if $u_j=u_{n-j+1}$ for $1\leq j\leq \floor {\frac{n}{2}}$. Every column is a palindrome if and only if $w_i=w_{m-i+1}$ for $1\leq i\leq \floor {\frac{m}{2}}$. 
\end{proof}
One can easily observe that the above result does not hold true in the case of $2$D palindromes that are not HV-palindromes. For example, the word $w$ given in Example \ref{e1} does not satisfy the conditions given in Proposition \ref{1}.

It is well known that a 1D word $u$ is a $1$D-palindrome if and only if $u = (xy)^ix$ for some 1D palindromes $x, y \in \Sigma^*$ and $i \ge 1$. In general, it was proved in \cite{MR2658256}, that if $\theta$ is an antimorphic involution, then $w\in P_{\theta}$ iff $w = \alpha(\beta \alpha)^i$ for some
$a, b \in P_{\theta}$ and $i\geq 0$ where $P_{\theta}$ denote the set of all
$\theta$-palindromes. For $2$D palindromic words, the condition was shown to be only sufficient in \cite{ms}.
We now show that the condition is both necessary and sufficient for an HV-palindrome.

\begin{proposition}\label{2}
Let $x,y$ be HV- palindromes. Then, $u=(x \obar y)^{i \obar} \obar x$ or $u=(x \ominus y)^{i \ominus} \ominus x$, $i \geq 1$ if and only if  $u$ is an HV- palindrome.
\end{proposition}
\begin{proof} 
Let $u=(x \obar y)^{i \obar} \obar x$, then 
$$
u^{R}=[(x \obar y)^{i \obar} \obar x]^{R} 
=x^{R} \obar (y^{R} \obar x^{R})^{i \obar} 
=(x^{R} \obar y^{R})^{i \obar} \obar x^{R} 
=(x \obar y)^{i \obar} \obar x 
=u
$$
Hence, $u$ is a $2$D palindrome. A similar proof follows for $u=(x \ominus y)^{i \ominus} \ominus x$. Now as $x$ and $y$ are HV-palindromes, then $x\obar y \obar x$ is an HV-palindrome. Hence, $u$ is an HV-palindrome for $i \geq 1$. \\Conversely, let $u=u_1\ominus u_2\ominus \cdots \ominus u_r$ be an HV-palindrome of size $(r,s)$, then by Proposition \ref{1}, $u_i=u_{r-i+1}$ for $1\leq i\leq \floor {\frac{r}{2}}$ and every row of $u$ is a palindrome. Let $x=u_1$ and $y=u_2\ominus u_3 \ominus \cdots \ominus u_{r-1}$. By Lemma \ref{l3}, $y$ is an HV-palindrome. Thus, $u = (x\ominus y)\ominus x$. 
\end{proof}
All letters from the alphabet $\Sigma$ are trivially palindromes and hence every $2$D word has palindromic sub-words. Also, all trivial palindromes are HV-palindromes and all HV-palindromes are $2$D palindromes.  However, a $2$D palindrome may or may not be an HV-palindrome. A word may or may not contain non-HV palindromes. In the following, we give a characterization of $2$D words such that all of its palindromic sub-words are HV-palindromes. The word $u$ in Example \ref{e2} is one such word.
\begin{theorem}
All $2$D palindromic sub-words of a $2$D word $w$ are HV-palindromes if and only if $w$ has no sub-word of the form 
$$\begin{matrix}
x&u&y\\v&p&v^R\\y&u^R&x
\end{matrix} $$
where $x,y\in \Sigma$ such that $x\neq y$, $u$ and $v^T$ are $1$D words and $p$ is a $2$D palindrome.
\end{theorem}
\begin{proof}
Let $w$ be a $2$D word  with no sub-word of the form $$\begin{matrix}
x&u&y\\v&p&v^R\\y&u^R&x
\end{matrix}$$ where $ x\neq y$ such that $x,y\in \Sigma$, $u$ and $v^T$ are $1$D words and $p$ is a $2$D palindrome. We show that all of its $2$D palindromic sub-words are HV-palindromes. Let $p'$ be a $2$D palindromic sub-word of $w$ of size $(r,s)$, with $r,s\ge 2$, then it must be of the form  $$\begin{matrix}
x_1&u_1&x_1\\v_1&p_1&v_1^R\\x_1&u_1^R&x_1
\end{matrix}$$ where $x_1\in \Sigma$, $u_1$ and $v_1$ are $1$D words and $p_1$ is a $2$D palindrome. We show that $p'$ is an HV-palindrome. Consider the first and the last row of $p'$. We show that they are the same. If not, consider the first position where they are different, say it is the  $i^{th}$ position. Let the $i^{th}$ position of the first and last row of $p'$ be $a$ and $b$ respectively, where $a\neq b \in \Sigma$. As $p'$ is a $2$D palindrome, then the $(s-i+1)^{th}$ position of the first and last row of $p'$ are $b$ and $a$ respectively. Then $p'$ has a sub-word $\begin{matrix}
p'_{1,i}&\cdots &p'_{1,s-i+1}\\
\vdots&\ddots&\vdots\\
p'_{r,i}&\cdots &p'_{r,s-i+1}
\end{matrix}$ of the form $\begin{matrix}
a&u_2&b\\v_2&p_2&v_2^R\\b&u_2^R&a
\end{matrix}$ where $ a\neq b$, $u_2$ and $v_2^T$ are $1$D words and $p_2$ is a $2$D palindrome which is a contradiction. Hence, the first row and the last row of $p'$ are same. Similarly, we can show that the $j^{th}$ and $(r-j+1)^{th}$ row of $p'$ are same for  $1\leq j\leq \floor {\frac{r}{2}}$. Now, consider the word $(p')^T$ which is a palindrome of size $(s,r)$. Apply the same procedure on $(p')^T$ to show that the  $j^{th}$ and ${(s-j+1)}^{th}$ row of $(p')^T$ are same for $1\leq j\leq \floor {\frac{s}{2}}$. This implies  $j^{th}$ and ${(s-j+1)}^{th}$ column of $p'$ is same for  $1\leq j\leq \floor {\frac{s}{2}}$. Thus, by Proposition \ref{1}, $p'$ is an HV-palindrome.\\
Conversely, if all palindromic sub-words of $w$ are HV-palindromes, then any sub-word of the form $$\begin{matrix}
x&u&y\\v&p&v^R\\y&u^R&x
\end{matrix}$$ where $x,y\in \Sigma$ are distinct, $u$ and $v^T$ are $1$D words and $p$ is a $2$D palindrome is itself a non-HV-palindrome which is a contradiction. 
\end{proof}

\subsection{Borders in a 2D palindrome}
In \cite{MR3539141}, the behaviour of borders of $1$D words was investigated and was shown that the asymptotic probability that a random word has a given maximal
border length $k$ is a constant, depending only on k and the alphabet size. The concept of $\theta$-bordered and $\theta$-unbordered words for an arbitrary morphic or antimorphic involution $\theta$ was reviewed in \cite{watson}. Here, we analyze the structure of borders in a $2$D palindrome. 
\begin{proposition}
Every border of a $2$D palindrome is a $2$D palindrome.
\end{proposition}
\begin{proof}
Let $w=[w_{i,j}]$ be a $2$D palindrome of size $(m,n)$. Let $v$ be a border of size $(i,j)$, then
\begin{equation}\label{q1}
    v=\begin{matrix}
w_{1,1} & w_{1,2} & \cdots  & w_{1,j} \\
w_{2,1} & w_{2,2} & \cdots  & w_{2,j} \\
\vdots & \vdots & \ddots & \vdots \\
w_{i,1} & w_{i,2} & \cdots  & w_{i,j}
\end{matrix} = \begin{matrix}
w_{m-i+1,n-j+1} & w_{m-i+1,n-j+2} & \cdots  & w_{m-i+1,n} \\
w_{m-i+2,n-j+1} & w_{m-i+2,n-j+2} & \cdots  & w_{m-i+2,n} \\
\vdots & \vdots & \ddots & \vdots \\
w_{m,n-j+1} & w_{m,n-j+2} & \cdots  & w_{m,n} 
\end{matrix}
\end{equation}
Now, as $w$ is a palindrome, then for $0\leq k\leq i-1$
\begin{equation}\label{q2}
    w_{1+k,1}  w_{1+k,2}  \cdots   w_{1+k,j}=(w_{m-k,n-j+1} w_{m-k,n-j+2}  \cdots   w_{m-k,n})^R 
\end{equation}
Hence, by Equations (\ref{q1}) and (\ref{q2}), $v$ is a $2$D palindrome.
\end{proof}
\begin{corollary}
Every border of an HV-palindrome is a $2$D palindrome.
\end{corollary}

Note that the border of an HV-palindrome need not be an HV palindrome. For example, in the word $aba\ominus bab\ominus aba$, the sub-array $ab\ominus ba $ is a border but is not an HV-palindrome.

 It is clear that the maximum number of borders in a word of size $(m,n)$ is $mn$ and is achieved when $|Alph(w)|=1$. Let $BOR(w)$ be the set of all borders of $w$, then we have the following result.
\begin{lemma}
Let $w$ be a word of size $(m,n),\;m,n\ge 2$.
\begin{enumerate}
 \item If $w$ is a $2$D palindrome, then $2\leq |BOR(w)|\leq mn$.
    \item If $w$ is an HV-palindrome, then $4\leq |BOR(w)|\leq mn$.
\end{enumerate}
\end{lemma}
\begin{proof}
Let $w$ be a word of size $(m,n)$, $m,n\ge 2$.
\begin{enumerate}
 \item If $w$ is a $2$D palindrome, then the prefix of size $(1,1)$ and $(m,n)$ is a border of $w$. The word $abb\ominus bbb\ominus bba$ achieves the lower bound. 
 \item If $w$ is an HV-palindrome, then the prefixes of size $(1,1)$, $(1,n)$, $(m,1)$ and $(m,n)$ are borders of $w$. The word $abba\ominus bbbb\ominus abba$ has only four  borders.
 \end{enumerate}
\end{proof}
\section{Counting Palindromes}
It can be easily observed that given an alphabet of size $q$, there are $q^{\ceil{\frac{n}{2}}}$ distinct $1$D palindromes of length $n$. In this section, we first count the number of distinct $2$D palindromes (HV-palindromes) that can be obtained over a given alphabet $\Sigma$. We then determine the number of $2$D palindromes (HV-palindromes) that can appear in the conjugacy class of a given $2$D word.

\begin{theorem}
Let $\Sigma$ be a finite alphabet such that $|\Sigma| =q $. Then, there are
\begin{enumerate}
    \item $q^i$, where $i = \ceil{\frac{mn}{2}}$ distinct $2$D palindromes of size $(m,n)$.

    \item $q^j$, where $j= \ceil{\frac{m}{2}}\ceil{\frac{n}{2}}$ distinct HV-palindromes of size $(m,n)$.
\end{enumerate}
\end{theorem}
\begin{proof}
Let $w$ be a $2$D palindrome of size $(m,n)$ such that $w= w_1\ominus w_2\ominus \cdots \ominus w_m$, where $w_i$ is the $i^{th}$ row of $w$. Since $w=w^R$, the $1$D word $u= w_1w_2\cdots w_m$ is a $1$D palindrome. Let $v$ be the prefix of $u$ of length 
 $\ceil{\frac{mn}{2}}$. Then, $v^R$ is a suffix of $u$. Thus, we have $\ceil{\frac{mn}{2}}$ distinct choices of letters from $\Sigma$ and therefore, there are  $q^{\ceil{\frac{mn}{2}}}$ distinct $2$D palindromes of size $(m,n)$ over $\Sigma$. 
 If $w$ is an HV-palindrome of size $(m,n)$, then by Proposition \ref{1}, we have $\ceil{\frac{m}{2}}\ceil{\frac{n}{2}}$ distinct  choices of letters in the prefix of $w$ of size $(\ceil{\frac{m}{2}},\ceil{\frac{n}{2}})$. Hence, there are  $q^{\ceil{\frac{m}{2}}\ceil{\frac{n}{2}}}$ distinct HV-palindromes of size $(m,n)$ over $\Sigma$.
\end{proof}
We illustrate with the following example.
\begin{example}
Consider the binary alphabet $\Sigma = \{a,b\}$. The set of all $2$D palindromes of size $(2,2)$ is
$$ \left\{ \begin{matrix}
aa\\aa\\
\end{matrix}, ~
~\begin{matrix}
ab\\ba\\
\end{matrix}, ~
~ \begin{matrix}
ba\\ab\\
\end{matrix}, ~
\begin{matrix}
bb\\bb\\
\end{matrix} ~\right\} $$
This set has four $2$D palindromes and two HV-palindromes.
\end{example}

\subsection{Palindromes in the conjugacy class of a word}

It was proved in \cite{2015arXiv150309112G} that the conjugacy class of a $1$D word contains at most two $1$D palindromes. The concept of $\theta$ conjugacy on words was defined in \cite{watson}. 

We consider the concept of conjugacy class in two-dimensional words and  count the number of $2$D palindromes and HV-palindromes. We recall the definition of conjugates of an array from \cite{smk}.

\begin{definition}
Let $u_1,u_2,\cdots u_m$ and $v_1,v_2,\cdots v_n$ be respectively the $m$ rows and $n$
columns of a word $w$ of size  $(m,n)$. The cyclic rotation of $k$ columns, for $1 \leq k \leq n$, denoted by $\circlearrowleft^{Col}_k$ is defined as the word

$$\circlearrowleft^{Col}_k w= v_{n-k+1}\obar\cdots \obar v_n\obar v_1\obar v_2 \obar \cdots \obar v_{n-k-1} \obar  v_{n-k}$$ 
Similarly, the cyclic rotation of $k$ rows, for $1 \leq k \leq m$, denoted by $\circlearrowleft^{Row}_k$
is defined
as the word $$\circlearrowleft^{Row}_k w = u_{m-k+1}\ominus \cdots \ominus u_m\ominus u_1\ominus u_2\ominus\cdots \ominus u_{m-k-1} \ominus  u_{m-k}$$
Then, the conjugacy class of $w$, denoted by
$Conj(w)$ is defined as
$$Conj(w) = \{ \circlearrowleft^{Col}_i \circlearrowleft^{Row}_j w, \;1 \leq i \leq n, \;1 \leq j \leq m \}$$
\end{definition}

Note that, given any $2$D word of size $(m,n)$, the number of elements in its conjugacy class can be at most $mn$. We illustrate with the following example.

\begin{example}\label{e7}
Consider the 2D word  $w= abc\ominus cbb\ominus bbc\ominus cba
$ of size $(4,3)$. Then, 
$$Conj(w)= \left\{~ \begin{matrix}
abc\\cbb\\bbc\\cba
\end{matrix},\begin{matrix}
cbb\\bbc\\cba\\abc
\end{matrix},\begin{matrix}
bbc\\cba\\abc\\cbb
\end{matrix},\begin{matrix}
cba\\abc\\cbb\\bbc
\end{matrix}, 
\begin{matrix}
bca\\bbc\\bcb\\bac
\end{matrix},\begin{matrix}
bbc\\bcb\\bac\\bca
\end{matrix},\begin{matrix}
bcb\\bac\\bca\\bbc
\end{matrix},\begin{matrix}
bac\\bca\\bbc\\bcb
\end{matrix}, 
\begin{matrix}
cab\\bcb\\cbb\\acb
\end{matrix},\begin{matrix}
bcb\\cbb\\acb\\cab
\end{matrix},\begin{matrix}
cbb\\acb\\cab\\bcb
\end{matrix},\begin{matrix}
acb\\cab\\bcb\\cbb
\end{matrix} 
~\right\}$$
\end{example}
\begin{remark}
For a 2D word $w$ of size $(m,n)$, if $|Alph(w)|=1$, $|Conj(w)|=1$ and if $|Alph(w)|=mn$, $|Conj(w)|=mn$. However, the converse need not be true as illustrated in Example \ref{e7}. 
If $w$ is a $2$D palindrome  with $|Alph(w)| = \ceil{\frac{mn}{2}}$ or an  HV-palindrome with $|Alph(w)| = \ceil{\frac{m}{2}}\ceil{\frac{n}{2}}$, then $|Conj(w)|=mn$. 
\end{remark}
We count the maximum number of $2$D palindromes (HV-palindromes) in the conjugacy class of a $2$D word. We call such conjugates as palindromic (HV-palindromic) conjugates and denote them by $PALConj(w)\; (HVPALConj(w))$. Note that, for the $w$ given in Example \ref{e7} we have, $$PALConj(w)=\left\{~ \begin{matrix}
abc\\cbb\\bbc\\cba
\end{matrix}, \begin{matrix}
bbc\\cba\\abc\\cbb
\end{matrix}
~\right\},~\; HVPALConj(w)=\emptyset$$
We give another example of a $2$D word $v=abba\ominus aaaa\ominus aaaa\ominus abba$ of size $(4,4)$ where $|HVPALConj(v)|= |PALConj(v)|= 4$. 
$$
PALConj(v)=HVPALConj(v)=\left\{~ \begin{matrix}
abba\\aaaa\\aaaa\\abba
\end{matrix}, \begin{matrix}
baab\\aaaa\\aaaa\\baab
\end{matrix},\begin{matrix}
aaaa\\abba\\abba\\aaaa
\end{matrix},\begin{matrix}
aaaa\\baab\\baab\\aaaa
\end{matrix}
~\right\}$$
 We first recall the following result for $1$D words from \cite{2015arXiv150309112G}.
\begin{theorem}\label{tt1}
 A conjugacy class of a $1$D word contains at most two palindromes and it has exactly two if and only if it contains a word of the form $(uu^R)^i$, where $uu^R$ is a primitive word and $i\geq 1$.
\end{theorem}
Consider a $2$D word $w= v_1\obar v_2\obar \cdots \obar v_n\;=\;  u_1\ominus u_2\ominus \cdots \ominus u_m$, where $v_i$ is the $i^{th}$ column and $u_j$ is the $j^{th}$ row of $w$. The word $w$ can be considered as a $1$D word over the alphabet of columns i.e., the set $A=\{ v_1,v_2, \ldots ,v_n\}$ and over the alphabet of rows i.e., the set $B=\{ u_1,u_2, \ldots ,u_m\}$. 
\begin{remark}\label{r88}
Note that a $2$D word is a $2$D palindrome if and only if it is a $1$D palindrome over its alphabet of columns and over its alphabet of rows. 
\end{remark}
We now have the following result.
\begin{theorem}\label{t51}
Let $w$ be a $2$D word of size $(m,n)$, then  
 \[0 \leq |PALConj(w)| \leq \begin{cases}
 4 , &\text{if m,\;n are even},\\
 1 , &\text{if m,\;n are odd},\\
 2 , &\text{otherwise}.
\end{cases}  \]

\end{theorem}
\begin{proof}
It is clear that there exist words $($for example, $w= aa\ominus ab)$ with no palindromic conjugates. Also, note that for a $2$D word $w$ of size $(m,n)$ such that $|Alph(w)|=1$, $|PALConj(w)| = 1$ and for $w$, such that $|Alph(w)|=mn$, $|PALConj(w)| = 0$.

We now find the number of palindromic conjugates in a word $w$ of size $(m,n)$. If $w$ has no palindromic conjugates, then we are done. Otherwise, assume that $v\in PALConj(w)$. Note that, $PALConj(w)=PALConj(v)$. Now, $v= v_1\obar v_2\obar \cdots \obar v_n\;=\;  u_1\ominus u_2\ominus \cdots \ominus u_m$, where $v_i$ is the $i^{th}$ column of $v$ and $u_j$ is the $j^{th}$ row of $v$.  As $v$ is a $2$D palindrome, by Remark \ref{r88}, it is a $1$D palindrome over the alphabet of columns and over the alphabet of rows.  We have the following cases.
\begin{itemize}
\item  [\textit{Case 1:}] If $m$ and $n$ are both even, then by Theorem  \ref{tt1}, as $n$ is even, there can be at most two $1$D palindromic conjugates of $v$ say $v$ and $v'$ over its alphabet of columns. Each of them can be expressed as a $1$D word over its alphabet of rows say the set $B$. Again, as $m$ is even, by Theorem \ref{tt1}, there are at most two palindromic conjugates of $v$ and $v'$ each as a $1$D word over $B$. In total, there are at most four palindromic conjugates in the conjugacy class of $w$.
Hence, for some $i_1,\;i_2,\;j_1$ and $j_2$, $$PALConj(w)=\{\circlearrowleft^{Col}_{i_1} \circlearrowleft^{Row}_{j_1} w, \circlearrowleft^{Col}_{i_1} \circlearrowleft^{Row}_{j_2} w,
\circlearrowleft^{Col}_{i_2} \circlearrowleft^{Row}_{j_1} w,
\circlearrowleft^{Col}_{i_2} \circlearrowleft^{Row}_{j_2} w\}$$ 
    \item [\textit{Case 2:}] If $m$ is odd and $n$ is even, then $v$ can be expressed as a $1$D word over its alphabet of rows say the set $B$. As $m$ is odd, by Theorem \ref{tt1}, there is no palindromic conjugate of $v$, other than $v$ over $B$. Again, $v$ can be expressed as a $1$D word over the alphabet of columns say the set $A$ to obtain at most two palindromic conjugates over $A$.
    In total, there are at most two palindromic conjugates in the conjugacy class of $w$. Hence, for some
    $i_1,\;i_2$ and $j_1$, $$PALConj(w)=\{\circlearrowleft^{Col}_{i_1} \circlearrowleft^{Row}_{j_1} w, 
\circlearrowleft^{Col}_{i_2} \circlearrowleft^{Row}_{j_1} w
\}$$ 
Similar is the case when  $m$ is even and $n$ is odd.    
Here, for some  $i_1,\;j_1$ and $j_2$, $$PALConj(w)=\{\circlearrowleft^{Col}_{i_1} \circlearrowleft^{Row}_{j_1} w, \circlearrowleft^{Col}_{i_1} \circlearrowleft^{Row}_{j_2} w\}$$    
    \item [\textit{Case 3:}] If $m,\;n$ are both odd, by Theorem \ref{tt1}, there is no palindromic conjugate of $v$, other than $v$. Hence, for some $i_1$ and  $j_1$, $$PALConj(w)=\{\circlearrowleft^{Col}_{i_1} \circlearrowleft^{Row}_{j_1} w\}$$ 
\end{itemize}
\end{proof}
Theorem \ref{t51} also holds in the case of HV-palindromes. We have the following result.
\begin{corollary}
Let $w$ be a $2$D word of size $(m,n)$, then  
 \[0 \leq |HVPALConj(w)| \leq \begin{cases}
 4 , &\text{if m,\;n are even},\\
 1 , &\text{if m,\;n are odd},\\
 2 , &\text{otherwise}.
\end{cases}  \]
\end{corollary}
\begin{remark}
 We now find the exact structure of the words that achieve the above upper bound in the case of HV-palindromes.  Let $v\in HVPALConj(w)$, then by the structure of $v$, $i_1=0,j_1=0,i_2= \frac{n}{2}$ and $ j_2=\frac{m}{2} $.
\begin{enumerate}
    \item [\textit{Case 1:}] If $m$ and $n$ are even, then to have four distinct palindromic conjugates, we have the following:

    \begin{enumerate}
        \item  $v\;\neq\; \circlearrowleft^{Col}_{\frac{n}{2}}v$ and $\circlearrowleft^{Row}_{\frac{m}{2}}v \;\neq\; \circlearrowleft^{Col}_{\frac{n}{2}} \circlearrowleft^{Row}_{\frac{m}{2}}v\;\implies$ the prefix of size $(m,\frac{n}{2})$ is not an HV-palindrome. 
        \item  $v\;\neq\; \circlearrowleft^{Row}_{\frac{m}{2}}v$ and $\circlearrowleft^{Col}_{\frac{n}{2}}v\;\neq\; \circlearrowleft^{Col}_{\frac{n}{2}} \circlearrowleft^{Row}_{\frac{m}{2}}v\;\implies$ the prefix of size $(\frac{m}{2},n)$ is not an HV-palindrome.
        \item  $w\;\neq\; \circlearrowleft^{Col}_{\frac{n}{2}} \circlearrowleft^{Row}_{\frac{m}{2}}v\;\implies$  the  prefix of size $(\frac{m}{2}, \frac{n}{2})$ is not a $2$D palindrome. 
        \item  $\circlearrowleft^{Col}_{\frac{n}{2}}v \;\neq\; \circlearrowleft^{Row}_{\frac{m}{2}}v\;\implies$ the prefix of size $(\frac{m}{2}, \frac{n}{2})$ is not a $2$D palindrome. 
       
    \end{enumerate}
Hence, in this case  $|PALConj(v)|=4$ iff the prefix of size $(\frac{m}{2}, \frac{n}{2})$ of $v$ is not a $2$D palindrome.
\item [\textit{Case 2:}] For $m$  even and $n$  odd, $|PALConj(v)|=2$ iff the prefix of size $(\frac{m}{2},n)$ is not an HV-palindrome.
        \item [\textit{Case 3:}] For $m$ odd and $n$  even, $|PALConj(v)|=2$  iff the prefix of size $(m,\frac{n}{2})$ is not an HV-palindrome.
 \end{enumerate}
\end{remark}

\section{Bounds on the number of palindromes}
In this section, we find the maximum and the least number of HV-palindromes in any $2$D finite and infinite word respectively. We also compare these results with the existing results on $2$D palindromes. 
\subsection{On the maximum number of HV-palindromes}
In this section, we find the maximum number of non-empty distinct HV-palindromes in a $2$D word $w$ of size $(m,n)$.  It was proved in \cite{pal3} and \cite{max} that there are at most $n$ palindromes in a $1$D word of length $n$ and  at most $2n+\floor{\frac{n}{2}}-1$ palindromes in a two-row array of size $(2,n)$ respectively. Further, it was conjectured in \cite{mc21} that the number of HV-palindromes in any $2$D word of size $(2,n)$ is less than or equal to $2n$. We give a proof of this conjecture. 
\begin{theorem}\label{t1}
The maximum number of HV-palindromes in any $2$D word of size $(2,n)$ is $2n, \; n\geq 1$.
\end{theorem}
\begin{proof}
Let $w$ be a $2$D word of size $(2,n)$. The number of HV-palindromes in $w$ is the sum of horizontal palindromes in $w$ and  HV-palindromes of size $(2,t),\;t\geq 1,$ in $w$. It can be observed that the HV-palindromes of size $(2,t),\;t\geq 1,$ are of the form $p\ominus p$, where $p$ is a $1$D palindrome. Let the number of the HV-palindromes of size $(2,t),\;t\geq 1,$ in $w$ be $k$. Then, there are $k$ palindromes of the form $p_i\ominus p_i$ where $1\leq i\leq k$. This implies there are $k$ common horizontal palindromes $p_i$ in both the rows of $w$ and hence, these $k$ horizontal palindromes should be counted only once. Thus, as only one horizontal palindrome can be created on the concatenation of a letter, then the maximum number of horizontal palindromes in $w$ is $\leq 2n-k$. Hence, the total number of HV-palindromes in $w$ $\leq 2n-k+k=2n$.
\end{proof}
We now give an example of a word of size $(2,n)$ that achieves this bound. Let $w= (a\ominus a)^{n\obar}$. It has $n$ horizontal palindromes: $a^i$ for $1\leq i\leq n$ and $n$ HV- palindromes: $(a\ominus a)^{j\obar}$ of size $(2,j)$ for $1\leq j\leq n$.
We deduce the following from Theorem \ref{t1}.
\begin{corollary}\label{c1}
The maximum number of HV-palindromes in any $2$D word of size $(m,2)$ is $2m, \; m\geq 1$.
\end{corollary}
We recall the following from \cite{max}.
\begin{lemma}\label{10}
Let $w$ be a $2D$ word of size $(m,n)$ for $ m,n \geq 2$. Then, the column concatenation of $w$ and $(x_1 \ominus x_2 \ominus\cdots \ominus x_m) $, where $  x_i\; \in \Sigma $ for each $i$ creates at most one distinct palindrome of size $(m,t)$ for some $t \geq 1$. Hence, 
there can be at most $\sum_{i=1}^{m}(m-i+1)=\frac{m(m+1)}{2}$  new palindromes created on the concatenation of a column to a word in $\Sigma^{m\times n}$.
Note that, this also holds true in the case of  HV-palindromes.
\end{lemma}
As $P_{2d}(w)=P_{2d}(w^T)$, we deduce the following result from Lemma \ref{10}.
\begin{corollary}\label{101}
Let $w$ be a $2D$ word of size $(m,n)$ for $ m,n \geq 2$. Then, the row concatenation of $w$ and $(x_1 \obar x_2 \obar \cdots \obar x_n) $, where $  x_i\; \in \Sigma $ for each $i$ creates at most one distinct palindrome of size $(t,n)$ for some $t \geq 1$.
The result also holds true in the case of HV-palindromes.
\end{corollary}
We generalize Theorem \ref{t1} to  words of larger sizes.
\begin{theorem}\label{t66}
The upper bound of the number of HV-palindromes in any $2$D word of size $(m,n)$ for $m,\; n\geq 2$ is 
\[  a_n = \left\{
\begin{array}{llll}
     \frac{m}{2}(( \frac{m}{2}+1)n-m+2), & & & if\; m \;is \;even\\
       (n-2)(\frac{m+1}{2})^2 +2m, & & & if\; m\; is \;odd.
\end{array} 
\right. \]
\end{theorem}
\begin{proof}
We prove the result by induction on $n$. The base case for $n=2$ is clear from Corollary \ref{c1}. Assume the result to be true for  $n=k-1$. Let $w$ be a word of size $(m,k)$.  Then by induction, the number of HV-palindromes in the prefix of size $(m,k-1)$ is bounded above by $a_{k-1}$. On the concatenation of the $k^{th}$ column, new palindromes of size $(i,t),\; t\geq 1$ for $1\leq i\leq m$ can be created. 

Consider any prefix of $w$ of size $(i,k)$ say $R_i$ for $1\leq i\leq m$. As a suffix of these $R_i$, by Lemma \ref{10}, only one extra HV-palindrome of size $(r,t)$ for each $r$ such that $1\leq r \leq i$ can be created.
We observe that as a suffix of each of these $R_i$, an HV-palindrome of size $(r,t)$ and an HV-palindrome of size $(i-r+1,t)$
for $t\geq 1$ and $1\leq r\leq  \ceil{\frac{i}{2}}$ can not be both newly formed as then the HV-palindrome with less number of rows will be present as the sub-word of the other one. Hence,  there can be at most $\sum_{i=1}^m\ceil{\frac{i}{2}}$ new HV-palindromes  formed on the concatenation of the $k^{th}$ column. Thus, the maximum number of HV-palindromes in any $2$D word of size $(k,n)$  is $$a_k=a_{k-1}+  \sum_{i=1}^m\ceil{\frac{i}{2}} \implies a_k=a_2+  (k-2)\sum_{i=1}^m\ceil{\frac{l}{2}} $$ Solving the recurrence relation for even and odd values of $m$ and with initial condition $a_2=2m$, we have 
\[  a_k = \left\{
\begin{array}{llll}
     \frac{m}{2}(( \frac{m}{2}+1)k-m+2), & & & if\; m \;is \;even\\
       (k-2)(\frac{m+1}{2})^2 +2m, & & & if\; m\; is \;odd.
\end{array} 
\right. \]
Hence, the result holds for $m=k$ and thus, the result follows by induction for $m\geq 2$.
\end{proof}
Table $1$ depicts the maximum number of distinct non-empty HV-palindromic sub-arrays in any binary word of size $(m,n)$ for larger values of $m$ and $n$ obtained by a computer program along with our obtained upper bound.
\begin{table}
\begin{center}
\begin{tabular}{| c | c | c  | c | c | c |c|c| c|}
\hline
 $m\times n$ &  Max $(HV)$ & Our Bound& $m\times n$ &  Max $(HV)$ & Our Bound\\ [0.5ex]
\hline\hline
 3 $\times$ 2& 6 &6&  3 $\times$ 6& 20&22\\
\hline
   3 $\times$ 3& 10 &10& 4 $\times$ 2&8&8 \\ 
\hline  3 $\times$ 4& 13 &14 & 4 $\times$ 3& 13&14\\ 
\hline  3 $\times$ 5& 17 &18& 4 $\times$ 4& 19&20\\ 
\hline   
  
\end{tabular}
\caption{Max HV-palindromes in a binary word of size $(m, n)$}
    \label{table1}
    \end{center}
\end{table}
It can be observed that the upper bound obtained in Theorem \ref{t66} is close to the actual values.
\subsection{Palindromic sub-arrays in a 2D palindrome}
 It was proved in \cite{pal3} that the maximum number of palindromes in a $1$D word $w$ is $|w|$ and $a^{|w|}$ is an example of a  palindrome that achieves it.\\ In this section find the maximum number of palindromes in a $2$D palindrome and the maximum number of HV-palindromes in a HV-palindrome of size $(m,n)$. We recall the following result from \cite{max}.
\begin{theorem}\label{u}
In a word $w$ of size  $(2,n)$, 
\begin{enumerate}
    \item the maximum number of palindromes in $w$  is $2n+\floor{\frac{n}{2}}-1$.
    \item if $w$ is a palindrome, then the maximum number of palindromes in $w$ is  $2n$.
\end{enumerate}

\end{theorem}

 We first prove the following.
\begin{lemma}\label{l22}
The maximum number of palindromes in a palindrome of size $(3,n)$ is $3n+\floor{\frac{n}{2}}-1$.

\end{lemma}
\begin{proof}
If $w$ is a palindrome of size $(3, n)$, then $w=w_1\ominus w_2 \ominus w_1^R$,  where $w_2=w_2^R$. By Theorem \ref{u}, there are $2n+\floor{\frac{n}{2}}-1$ palindromes in $w_1\ominus w_2$. Since, $w$ is obtained by concatenating $w_1^R$ to $w_1\ominus w_2$ as the last row, the only new palindromes that are obtained are palindromes of size $(3,t)$, $t\ge 1$. Thus, we just count the palindromes formed by concatenating the last row.
By Corollary \ref{101}, at most $n$ such palindromes can be created. Hence, $P_{2d}(w)\leq n+ 2n+\floor{\frac{n}{2}}-1\;=\;3n+\floor{\frac{n}{2}}-1$.
\end{proof}
Using Theorem \ref{u} and Lemma \ref{l22}, we have the following result for the maximum number of palindromes in a palindrome of size $(m,n)$.

\begin{theorem} \label{t6}
The upper bound on the number of palindromes in a $2$D palindrome/HV-palindrome of size $(m,n)$ for $m,\; n\geq 2$ is $a_n$ which is one of the following.
\begin{center}

\begin{tabular}{|c|c|c|}
\hline
     $m$ & $n$ & \;$a_n$ \\
     \hline
    
     \;even\; & \;even\;& $2m+ (\frac{n-2}{2})(\frac{m(m+1)}{2} + \frac{m}{2}(\frac{m}{2}+1))$ \\
     \hline\;even\; &\;odd\;&$3m+(\frac{n-3}{2})(\frac{m(m+1)}{2} +  \frac{m}{2}(\frac{m}{2}+1)) +\floor{\frac{m}{2}}-1$
     \\
     \hline
        
\;odd\; &\;even\; &
     $2m+ (\frac{n-2}{2})(\frac{m(m+1)}{2} + (\frac{m+1}{2})^2)$\\\hline \;odd\; &\;odd\; & $3m+(\frac{n-3}{2})(\frac{m(m+1)}{2} + (\frac{m+1}{2})^2) +\floor{\frac{m}{2}}-1$ \\

      \hline
\end{tabular}
    
\end{center}
\end{theorem}
\begin{proof}
Let $w$ be a $2$D palindrome of size $(m,n)$. We prove the result using induction on $n$. The cases when $n=2$ and $n=3$ are clear by taking transpose of the words in Theorem \ref{u} and Lemma \ref{l22} respectively. Assume the result to be true for $n=k-1$. Let $w$ be a palindrome of size $(m,k)$. Then, $w$ is of the form $w_1\obar v \obar w_1^R$, where $w_1$ is of size $(m,1)$ and $v$ is a palindrome of size $(m,k-2)$.
Note that as $w$ is a palindrome and $v=v^R$, the  palindromes of size $(i, t),\;1\leq t\leq k-1$ for $1\leq i\leq m$ formed by concatenation of $w_1$ and $v$ are also formed on concatenation of $v^R$ and $w_1^R$.  Hence,  the number of palindromes in $w$ is the sum of palindromes in $v$, the palindromes formed on concatenation of $w_1$ to $v$ and the palindromes of size $(i,n)$ formed on concatenation of $w_1^R$ to $w_1\obar v$. By Lemma \ref{10}, at most $\frac{m(m+1)}{2}$ palindromes are formed on concatenation of $w_1$ to $v$. We just have to count the number of palindromes of size $(i,n)$ formed by concatenation of $w_1^R$ to $w_1\obar v$. As $w$ is a palindrome, the $i^{th}$ row of $w$ is the reverse of the $m-i+1^{th}$ row of $w$, so, there can be at most    $\ceil{\frac{i}{2}}$ distinct palindromes formed of size $(m-i+1,n) $ for $1\leq i\leq m$. Hence, $$a_k= a_{k-2} +\frac{m(m+1)}{2} +\sum_{i=1}^{m} \ceil{\frac{i}{2}}$$ Solving the recurrence relation, with initial conditions of $n=2$ and $n=3$ for even and odd values of $k$, we get the result.
\end{proof}
We now find the maximum number of HV-palindromes in an HV-palindrome of size $(m,n)$.
We first observe the following result by Theorem \ref{t1} and Corollary \ref{101}.
\begin{lemma}\label{l22h}
The maximum number of HV-palindromes in a HV-palindrome of size $(3,n)$ is $3n$. The word $ a^n\ominus b^n\ominus a^n$ proves that the bound is tight.
\end{lemma}
\begin{corollary}
 \label{t6h}
The upper bound on the number of HV-palindromes/palindromes in a HV-palindrome of size $(m,n)$ for $m,\; n\geq 2$ is $a_n$ which is one of the following.
\begin{center}
\begin{tabular}{|c|c|c|}
\hline
     $m$ & $n$ & \;$a_n$ \\
     \hline
    
     \;even\; & \;even\;& $2m+ (\frac{(n-2)m}{2})(\frac{m}{2}+1)$ \\
     \hline\;even\; &\;odd\;&$3m+ (\frac{(n-3)m}{2})(\frac{m}{2}+1)$
     \\
     \hline
        
\;odd\; &\;even\; &
     $2m+ (n-2)(\frac{m+1}{2})^2 $\\\hline \;odd\; &\;odd\; & $3m+ (n-3)(\frac{m+1}{2})^2$ \\

      \hline
\end{tabular}

\end{center}

\end{corollary}
\begin{proof}
Following the proofs of Theorems \ref{t66}  and \ref{t6}, the word $w$ of size $(k,n)$ is of the form $w_1\obar v\obar w_1$, where $v$ is an HV-palindrome of size $(k-2,n)$, there can be at most  $\sum_{l=1}^m\ceil{\frac{l}{2}}$ palindromes added by concatenating $w_1$ to $v$. We just have to count the palindromes of size $(i,n)$ for $1\leq i\leq m$ formed in  $w_1\obar v\obar w_1$. We get the following recurrence relation 
$$a_k=a_{k-2}+  \sum_{l=1}^m\ceil{\frac{l}{2}}  +\sum_{i=1}^{m} \ceil{\frac{i}{2}}=a_{k-2}+  2\sum_{l=1}^m\ceil{\frac{l}{2}}$$

which gives the result.
\end{proof}
 Note that the bounds in Corollary \ref{t6h} are equal to that of Theorem \ref{t66} for even values of $n$ and are less than that of Theorem \ref{t66} for odd values of $n$.
\subsection{On the least number of HV-palindromes}
In this section, we find the least number of non-empty distinct palindromes in a $2$D infinite word $w$ with $|Alph(w)|=q$.

A $1$D infinite word is an infinite sequence of symbols. It was proved in \cite{MR3037792} that there are at least $8$ palindromes in an infinite $1$D word. A $2$D infinite word is an array with infinite rows and columns. We recall the following result on the least number of $2$D palindromes in an infinite $2$D word from \cite{tcs}.
\begin{theorem}
The least number of $2$D palindromes in an infinite $2$D word is 
\[ \left\{ 
\begin{array}{llll}
      \infty , & & & if \;|Alph(w)|=1 \\
      20 , &&& if \;|Alph(w)|=2\\
      q, &&& if \;|Alph(w)|= q,\; where\; q\geq 3.
\end{array} 
\right. \]
\end{theorem}
We have the following result for HV-palindromes.
\begin{theorem}\label{t5}
The least number of HV-palindromes in an infinite $2$D word is 
\[ \left\{ 
\begin{array}{llll}
      \infty , & & & if \;|Alph(w)|=1 \\
      14 , &&& if \;|Alph(w)|=2\\
      q, &&& if \;|Alph(w)|= q,\; where\; q\geq 3.
\end{array} 
\right. \]
\end{theorem}
\begin{proof}
Let $w$ be a $2$D infinite word with $|Alph(w)|=q$. We have the following.
\begin{itemize}
    \item \textbf{Case 1:} If $|Alph(w)|=1$, then $w=(x^{\infty \ominus})^{\infty \obar}$, where $x\in \Sigma$. This word has infinite HV-palindromes: $(x^{i \ominus})^{j\obar}$ for $i,\;j\geq 1$.
     \item \textbf{Case 2:} If $|Alph(w)|=2$, then let $w$ be a binary word on $\Sigma=\{a,b\}$. It was shown in \cite{MR3037792}, that any finite $1$D binary word of length greater than $8$, has at least $8$ palindromic factors. All these palindromes are HV-palindromes. Since every infinite $1$D word must have at least $8$ HV-palindromes, every row and column of $w$ has at least $8$ HV-palindromes. The only palindromes that can be common to both are the trivial palindromes i.e. $a$ and $b$. Thus, $w$ has at least $6$ horizontal and $6$ vertical non-trivial HV-palindromes. Thus, any $2$D infinite binary word $w$ has at least $8+8-2=14$ HV-palindromes. We give an example of the word that achieves the bound. Let $u_1=ababba$ and $u_{i+1}$ be the $1$-cyclic shift of $u_{i}$ for $1\leq i\leq 5$ and $v=u_1\ominus u_2\ominus \cdots \ominus u_6$.
Note that, for the given $v$, $(v^{2\ominus})^{2\obar}$ has exactly $14$ HV-palindromes: $\{p,\;p^T : p\in A\}$ where $A=\{a,b,aa,bb, aba, bab, abba, baab\}$. As there is no palindrome of size $(i,1)$ or $(1,j)$ for $i,\;j\geq  5$ in $(v^{2\ominus})^{2\obar}$, thus, there are only $14$ HV-palindromes in $(v^{\infty \ominus})^{\infty \obar}$ which is the required word.
\item \textbf{Case 3:} If $|Alph(w)|=q$, where $q\geq 3$, then there are at least $q$ trivial HV-palindromes. We give a word with exactly $q$  HV-palindromes. Let $Alph(w)=\{a_1,a_2,\ldots, a_q\}$ and $$u=\begin{matrix}
a_1&a_2&\cdots& a_{q-1}&a_q\\
a_2&a_3&\cdots& a_{q}&a_1 \\
\vdots && \ddots&\vdots& \vdots \\
a_q&a_1&\cdots& a_{q-2}&a_{q-1}
\end{matrix} .$$
Then, the word $w = (u^{\infty\ominus})^{\infty\obar}$ has $q$ HV-palindromes: $\{a_1, a_2,\ldots ,a_q\}$
\end{itemize}
\end{proof}
As the palindromes when $q\geq 3$ in the above theorem are all trivial, we find the least number of HV-palindromes in an infinite $2$D word $w$ with at least one non-trivial HV-palindrome such that $|Alph(w)|\geq 3$. 
\begin{theorem}\label{t11}
The least number of HV-palindromes in an infinite $2$D word $w$ with $|Alph(w)|=q$ that has at least one non-trivial HV-palindrome is 
$\begin{cases}
        5, & \text{if $q=3$}, \\
        q+1, & \text{if $q>3$}.
         \end{cases}$
\end{theorem}
\begin{proof}
Let $w$ be a $2$D infinite word with $|Alph(w)|=q$. 
We have the following cases.
\begin{itemize}
    \item \textbf{Case 1:} If $|Alph(w)|=3$, then consider an infinite $2$D word with exactly one non-trivial HV-palindrome. We observe that this non-trivial HV-palindrome should occur as a sub-word and is of one of the forms $xx,xyx$ or their transpose where $x,y\in \Sigma$. If $|Alph(w)|=3$ i.e $\Sigma=\{a,b,c\}$, then we cannot construct a sub-word of size $(3,3)$ with the above as prefixes with only one non-trivial  HV-palindrome. Hence, there is no infinite $2$D word with exactly one non-trivial HV-palindrome when $|Alph(w)|=3$. Thus, we construct an infinite $2$D word $w$ with exactly two non-trivial HV-palindromes such that $|Alph(w)|=3$. Let $w= \left( a(abc)^{\obar \infty} \right)\ominus {\left( v^{ \obar\infty} \right)}^{\ominus \infty}$ for $v=bca\ominus abc\ominus cab$. It has only two non-trivial HV-palindromes $aa$ and $a\ominus b \ominus a$.  Thus, the least number of HV-palindromes in an infinite $2$D word $w$ with $|Alph(w)|=3$ that has at least one non-trivial HV-palindrome is $5$.
 \item \textbf{Case 2:} If $|Alph(w)|=q$, for $q\geq 4$, then the least number of HV-palindromes in an infinite $2$D word $w$ with  at least one non-trivial HV-palindrome should be greater than or equal to $q+1$. We give the existence of such a word with exactly $q+1$ HV-palindromes. Let $w= (u^{\infty\ominus})^{\infty\obar}$ for $$u=\begin{matrix}
a_1&a_2&a_3&\cdots& a_{q-2}&a_{q-1}\\
a_1&a_3&a_4&\cdots& a_{q-1}&a_{q}\\
a_2&a_4&a_5&\cdots& a_{q}&a_{1}\\
\vdots&\vdots &\vdots& \ddots& \vdots&\vdots \\
a_{q-2}&a_q&a_1&\cdots& a_{q-4}&a_{q-3}\\
a_{q-1}&a_1&a_2&\cdots &a_{q-3}&a_{q-2}
\end{matrix}$$
Here, $w$ has one non-trivial HV-palindrome $a_1\ominus a_1$ along with trivial palindromes.
\end{itemize}
\end{proof}
\section{Conclusions}\label{sec4}HV-palindromes is a special class of $2$D palindromes in which every row and column is a $1$D palindrome. We witnessed certain important combinatorial properties by investigating the structure of an HV-palindrome. We have an affirmative answer to the conjecture proposed in \cite {mc21} for the maximum number of HV-palindromes in a word of size $(2,n)$ and a generalization for  a word of size $(m,n),\; m\geq 2$. We also analyzed the least number of HV-palindromes in an infinite $2$D word. In  future, it will be interesting to study and compare the properties of $1$D palindromes and HV-palindromes/$2$D palindromes.

\bibliographystyle{abbrv}
\bibliography{jalc-ex.bib}

\end{document}